\RequirePackage{fix-cm}
\documentclass{svjour3}                     
\smartqed  
\usepackage{graphicx}
\usepackage{amsmath,amssymb,latexsym,bbm}
\begin{document}

\title{Diagrammatic proof of the large $N$ melonic dominance in the SYK model}

\titlerunning{SYK model - Diagrammatic proof of the large $N$ melonic dominance}

\author{V. Bonzom\thanks{V.~Bonzom is partially supported by the CNRS Infiniti ModTens grant and by the ANR MetAConC project ANR-15-CE40-0014.} \and V. Nador\thanks{V. Nador is fully supported by CNRS Infini ModTens grant.} \and A. Tanasa\thanks{A.~Tanasa is partially supported by the CNRS Infiniti ModTens grant.}}


\institute{V. Bonzom \at
              LIPN, UMR CNRS 7030, Institut Galil\'ee, Universit\'e Paris 13,
99, avenue Jean-Baptiste Cl\'ement, 93430 Villetaneuse, France, EU \\
              \email{bonzom@lipn.univ-paris13.fr}           
           \and
           V. Nador \at
              LaBRI, Universit\'e de Bordeaux, 351 cours de la Lib\'eration, 33405 Talence, France, EU
              \email{victor.nador@ens-lyon.fr}
           \and 
           A. Tanasa \at
              LaBRI, Universit\'e de Bordeaux, 351 cours de la Lib\'eration, 33405 Talence, France, EU
              \email{ntanasa@u-bordeaux.fr}
}

\date{Received: date / Accepted: date}

\maketitle

\begin{abstract}
A crucial result on the celebrated Sachdev-Ye-Kitaev model is that its large $N$ limit is dominated by melonic graphs. In this letter we offer a rigorous, diagrammatic proof of that result by direct, combinatorial analysis of its Feynman graphs.

\end{abstract}

\section*{Introduction} \label{intro}

The Sachdev-Ye model \cite{Sachdev_Ye} was initially introduced within a condensed matter perpsective in the early nineties. Kitaev~\cite{SY_Kitaev}, in a series of talks, exposed its connection to the celebrated AdS/CFT correspondence. The model is now known as the Sachdev-Ye-Kitaev (SYK) model.

It has attracted a huge deal of interest from both condensed matter, see for example \cite{georges}, and high energy physics communities, see for example \cite{SYK_MS}, \cite{SYK_PR}, \cite{SYK_GR}, \cite{BLT}, \cite{quenched}, \cite{Dario1}, \cite{Dario2} and references within.

A crucial property of the SYK model is that the large $N$ limit, where $N$ is the number of fermions, is dominated by the set of melonic graphs. Remarkably, those graphs had been known to dominate the large $N$ limit of random tensor models \cite{gurau-book}, $N$ being here the size of the tensor. This is true for the Gurau's colored tensor model \cite{bonzom}, Tanasa's multiorientable model \cite{mo}, \cite{tanasar}, Tanasa and Carrozza's model \cite{ct}. Even for tensor models whose large $N$ limit does no consist of melonic graphs, the universality class of melonic graphs is easily stumbled upon \cite{bonzomr}. 

The large $N$ melonic dominance is mentioned in Kitaev's original talk or in the seminal article \cite{SYK_MS}, to name only a few of the very well-known references mentioning it. It is a direct cause of interesting propoerties of the model, in particular the form of the Schwinger-Dyson equation for the 2-point function, which leads to conformal invariance in the IR.

The large $N$ dominance of melonic graphs in both the SYK model and tensor models triggered interesting developments, starting from the Gurau-Witten model \cite{GurauWitten} and the CTKT model \cite{CTKT}. This has motivated $1/N$ expansions for new tensorial models \cite{IrreducibleTensorsCarrozza}, \cite{SymmetricTraceless}, \cite{TensorialGrossNeveu}, \cite{TwoSymmetricTensors} and the new field of tensor quantum mechanics \cite{O4}, \cite{Sextic}, \cite{TensorQM}, \cite{SpectraTensorModels}.

A proposal which is very close to both the SYK model and unitary-invariant tensor models is the so-called colored SYK model, initially introduced in \cite{SYK_GR} and \cite{complete}, where the fermions come in several flavors, or colors, and only different colors can interact. In this model and in the Gurau-Witten model, combinatorial methods originating from the study of tensor models can be used, for instance to identify the Feynman graphs which contribute at a given order in the $1/N$ expansion \cite{complete}, \cite{BLT}.

However, combinatorial proofs have been of limited use so far in the original SYK model. One reason is that colors (or flavors) have been crucial to most combinatorial results in tensor models but the Feynman graphs of the SYK model have no colors (only recently new methods have been found to deal with tensor models without colors \cite{IrreducibleTensorsCarrozza}, \cite{SymmetricTraceless}, \cite{TensorialGrossNeveu}, \cite{TwoSymmetricTensors}, but have not been applied to the original SYK model to the best of our knowledge).

In this letter we give a mathematically rigorous proof the dominance of melonic graphs of the SYK model, Theorem \ref{thm}. One might expect it to be difficult without the notion of colors. As it turns out, even for graphs without colors, the set of melonic graphs is simple enough that our direct, combinatorial approach remains quite elementary. It revolves around the fact that ultimately, melonic graphs are characterized by their 2-cuts, so we can study the large $N$ behavior of Feynman graphs under 2-cuts.

We present the set of Feynman graphs to be studied in Section \ref{sec:DiagExpSYK}, then we intorduce melonic graphs and prove Theorem \ref{thm} in Section \ref{sec:Proof}.

%

\section{Diagrammatic expansion of the SYK model} \label{sec:DiagExpSYK}

\subsection{Definition of the SYK model}
\label{sec:Def}

The SYK model has $N$ Majorana fermions coupled via a $q$-body random interaction
\begin{equation}
H_{SYK} = i^{q/2} \sum_{i_1, \dotsc, i_q} J_{i_1\dotsb i_q} \psi_{i_1}(t) \dotsm \psi_{i_q}(t),
\label{H_SYK}
\end{equation}
where $J_{i_1\dotsb i_q}$ is the coupling constant. Quenching means that it is a random tensor with a Gaussian distribution such that 
\begin{equation} \label{Covariance}
\langle J_{i_1\dotsb i_q}\rangle=0 \qquad \text{and} \qquad \langle J_{i_1\dotsb i_q} J_{j_1\dotsb j_q}\rangle = 6 J^2 N^{-(q-1)} \prod_{m=1}^{q} \delta_{i_m,j_m}
\end{equation}
The fields $\psi_i(t)$ satisfy fermionic anticommutation relation $\{\psi_i(t),\psi_j(t)\} = \delta_{i,j}$.

\subsection{Stranded structure of the Feynman graphs}
\label{sec:Stranded}

When doing perturbation theory, the interaction term is represented by a vertex with $q$ incident fermionic lines. Each fermionic line $m=1\dotsc, q$ carries an index $i_m=1, \dotsc, N$ which is contracted at the vertex with a coupling constant $J_{i_1 \dotsb i_q}$. The free energy expands onto those connected, $q$-regular graphs.

Averaging over the disorder is performed using Wick contractions between the $J$s sitting on vertices, with covariance \eqref{Covariance}. Each vertex thus receives an additional line, which we represent as a dashed line and call a {\bf disorder line}. An example of such a Feynman graph of the SYK model is given  in Fig.~\ref{SYK_graph}.

\begin{figure}[ht]
\centering
\includegraphics[scale=0.65]{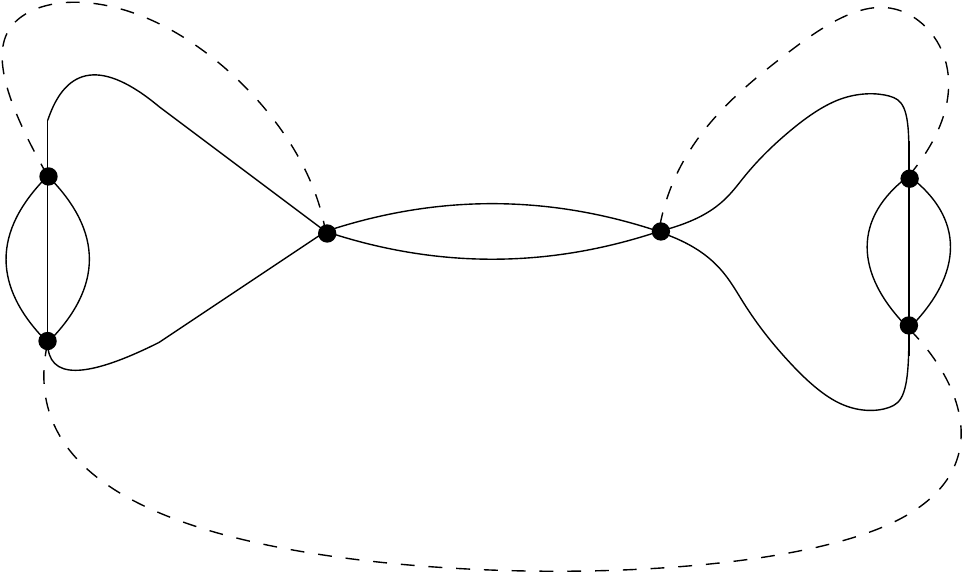}
\caption{An example of a Feynman graph of the $q=4$ SYK model.}
\label{SYK_graph}
\end{figure}

The above description of the Feynman graphs is however not enough to describe the $1/N$ expansion as it ignores the indices of the random couplings. Indeed a disorder line propagates in fact $q$ field indices, where the field index of fermionic line incident on a vertex is identified with the index of a fermionic line at another vertex. We thus have to represent a disorder line as a line made of $q$ strands where each strand connects fermionic lines as follows
\begin{equation}
\langle J_{i_1\dotsb i_q} J_{j_1\dotsb j_q}\rangle = \begin{array}{c} \includegraphics[scale=.4]{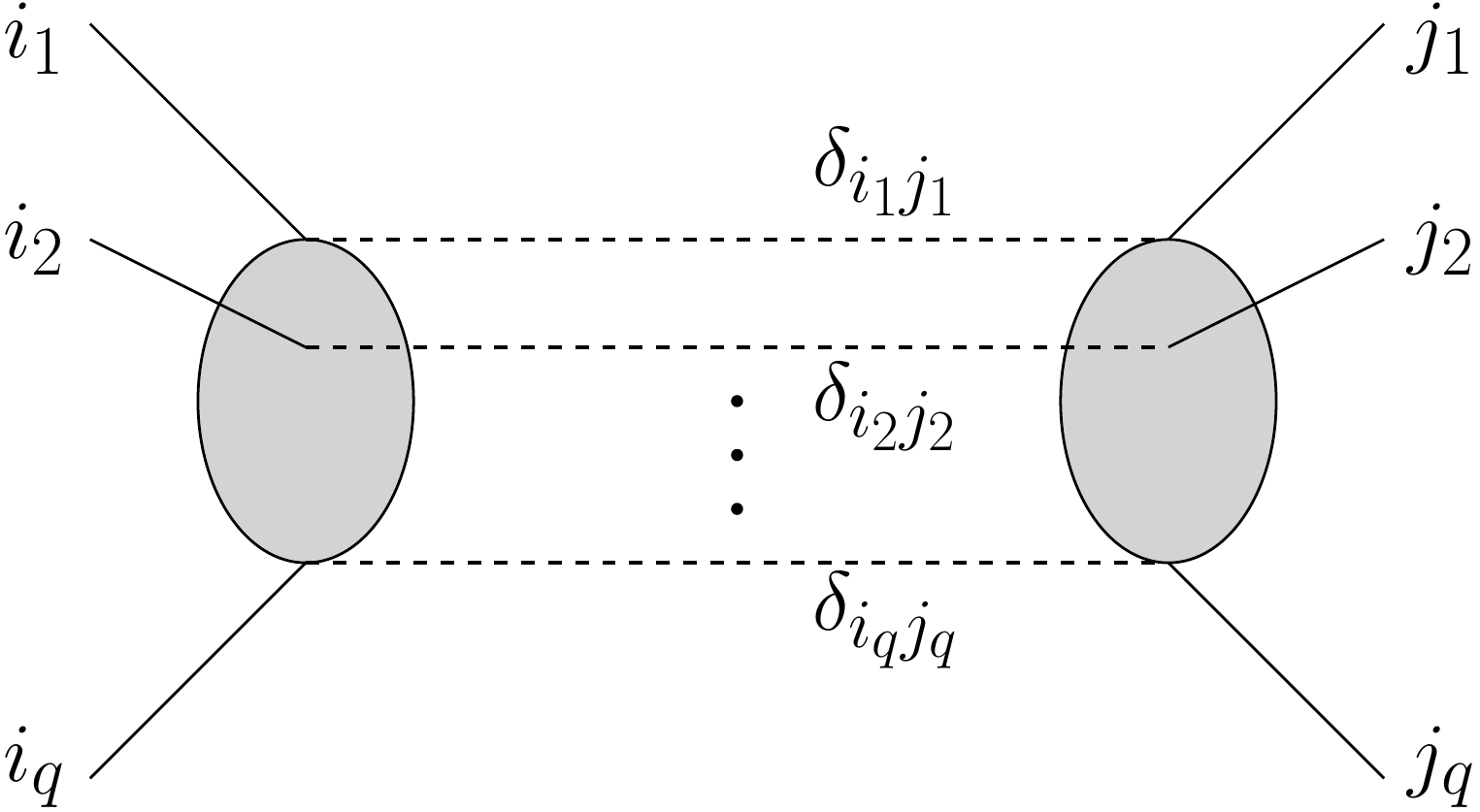} \end{array}
\end{equation}
Here the grey discs represent the Feynman vertices.

We denote $\mathbbm{G}$ the set of Feynman graphs of the SYK model with stranded disorder lines. For $G\in \mathbbm{G}$, we further denote $G_0\subset G$ the $q$-regular graph obtained by removing the strands of the disorder lines, see Fig.~\ref{SYK_graph0}. Due to the quenching averaging the fermionic free energy over the disorder, $G_0$ is connected. Moreover, due to the Wick contractions, each vertex has exactly one incident disorder line. This implies that $G$, hence $G_0$, have an even number of vertices.
 
\begin{figure}[ht]
\centering
\includegraphics[scale=0.75]{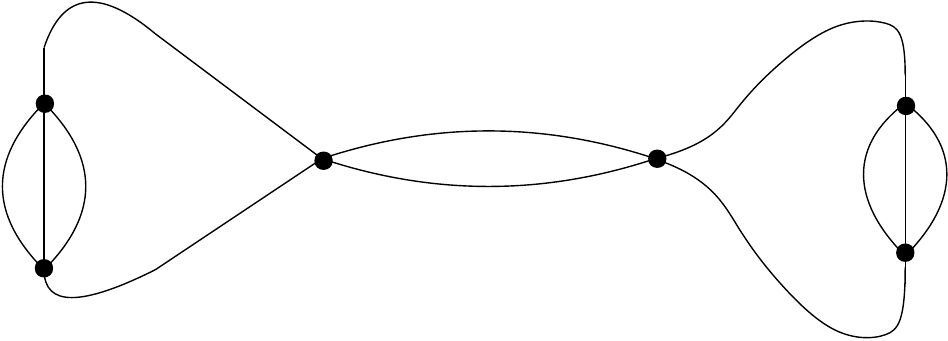}
\caption{The graph obtained after deleting of the disorder lines of the graph of Fig.~\ref{SYK_graph}.}
\label{SYK_graph0}
\end{figure}

There is a single graph with two vertices
\begin{equation}
G_{\min} = \begin{array}{c} \includegraphics[scale=.45]{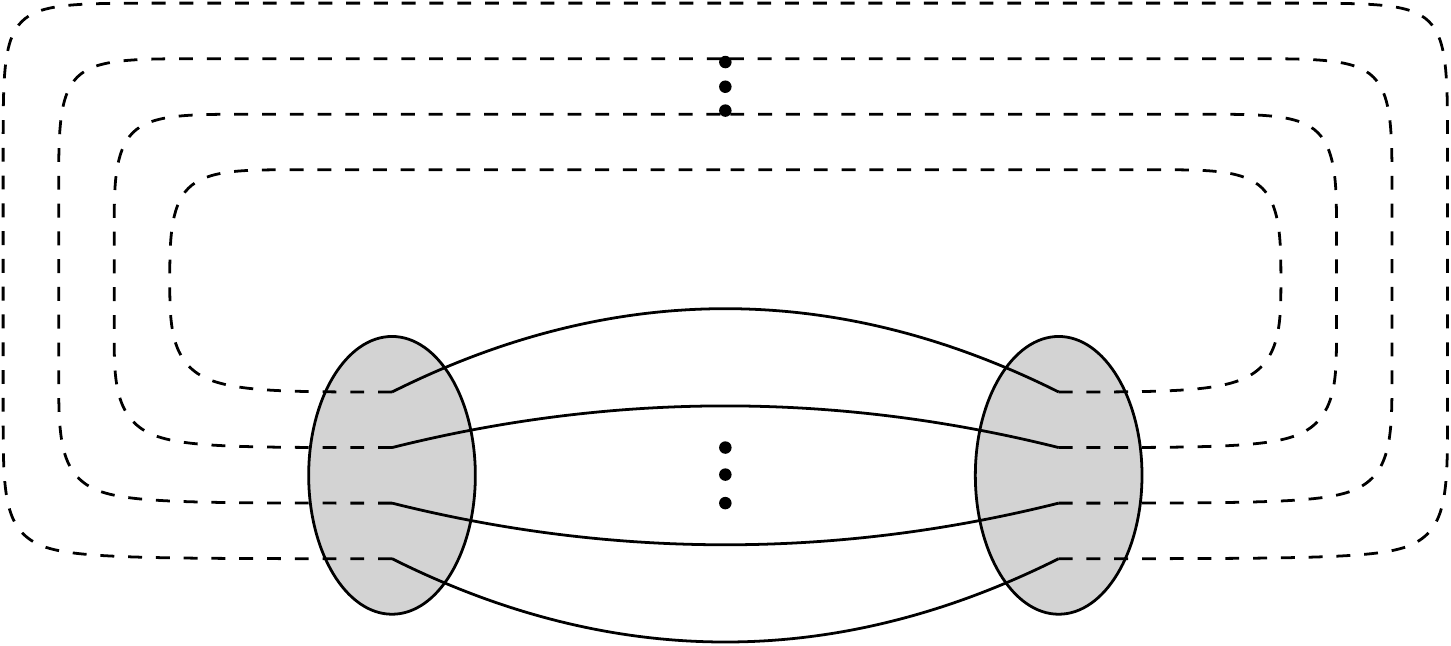} \end{array}
\end{equation}
and $G_{0,\min} = \begin{array}{c} \includegraphics[scale=.25]{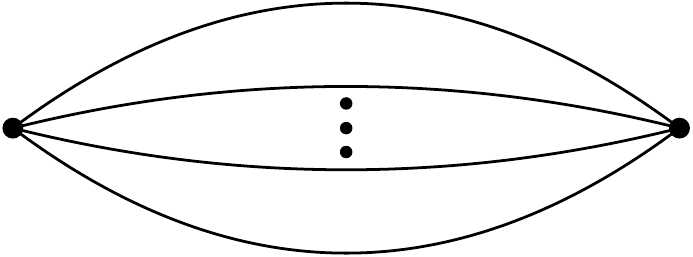} \end{array}$.

The field index is preserved along each fermionic line and along each strand of disorder lines. This means that there is a free sum for each cycle made of those lines, thereby contributing to a factor $N$.

\begin{definition}[Faces]
A cycle made of alternating fermionic lines and strands of disorder lines is called a face. We denote $F(G)$ the number of faces of $G\in\mathbbm{G}$.
\end{definition}

$G\in\mathbbm{G}$ thus receives a factor $N$ per face. It also receives a factor $N^{-(q-1)}$ for each disorder line. The weight of a graph is thus
\begin{equation} \label{scaling}
W(G) = N^{\delta(G)} \qquad \text{with} \qquad \delta(G) = F(G) - (q-1) V(G)/2
\end{equation}
where $V(G)$ is the number of vertices. To find the large $N$ limit, we can thus find the graphs which maximize the number of faces at fixed number of vertices.

\section{Proof of the melonic dominance in the large $N$ limit} \label{sec:Proof}

\subsection{Melonic graphs} \label{sec:Melons}

\begin{definition}[Melonic graphs]
A melonic move is the insertion, on a fermionic line, of the following 2-point graph
\begin{equation}
\begin{array}{c} \includegraphics[scale=.4]{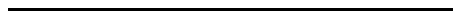} \end{array} \qquad \to \qquad \begin{array}{c} \includegraphics[scale=0.45]{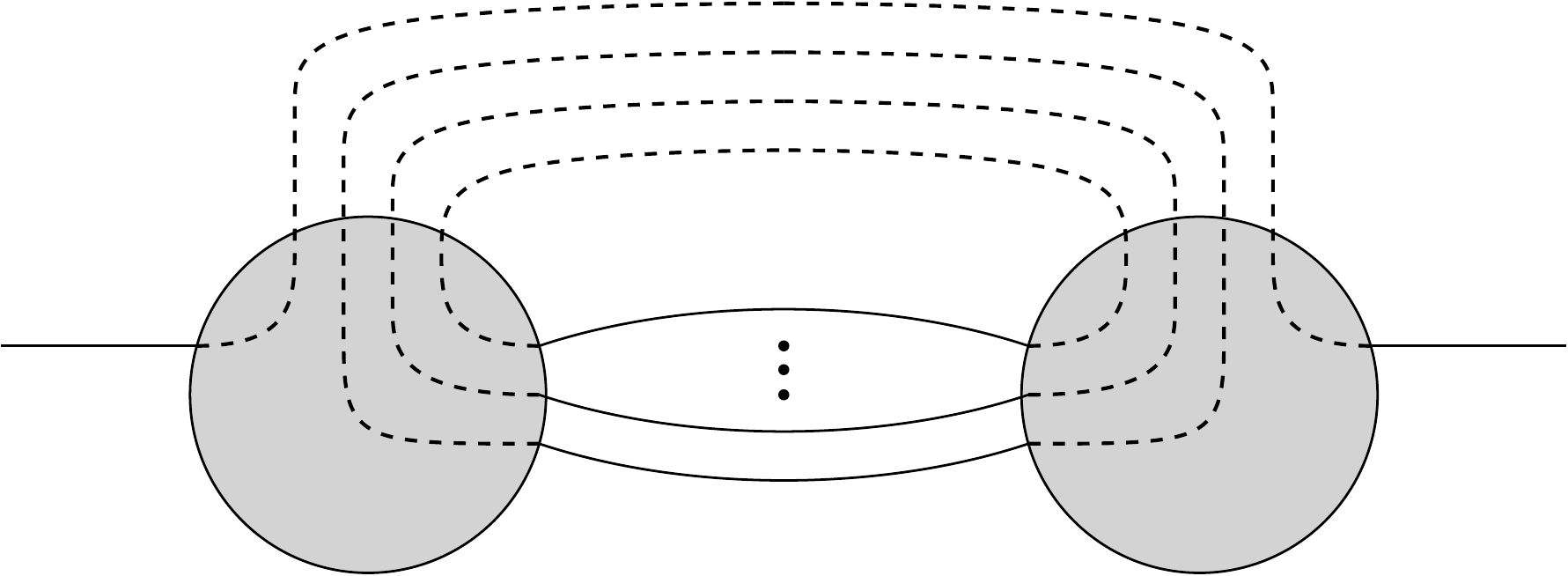} \end{array}
\end{equation}
A melonic graph is a graph obtained from $G_{\min}$ by repeated melonic moves. An example is provided in Fig.~\ref{SYK_melon_ex}.
\end{definition}

\begin{figure}[ht]
\centering
\includegraphics[scale=0.7]{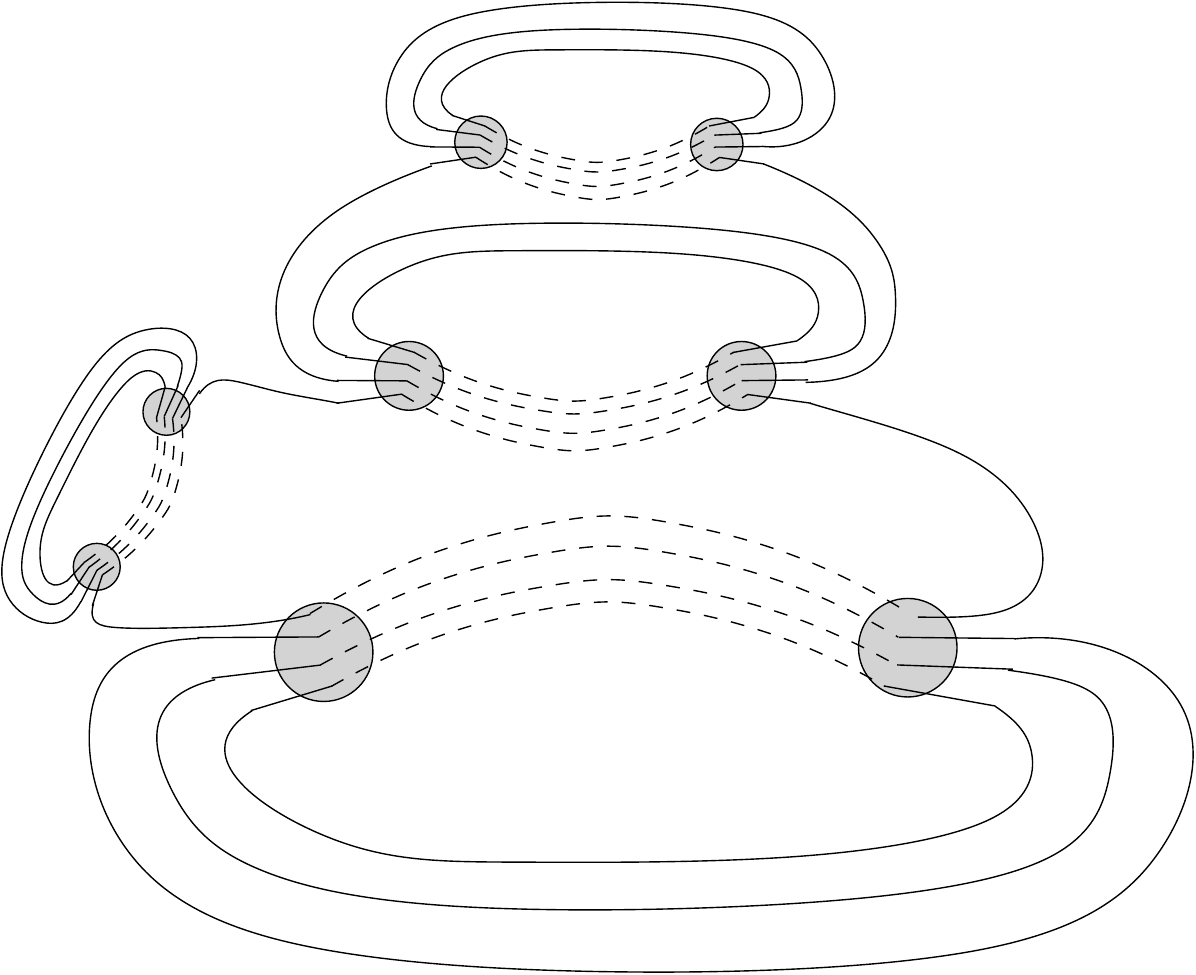}
\caption{An example of melonic graph}
\label{SYK_melon_ex}
\end{figure}

The number of faces of melonic graphs is easily found.

\begin{proposition} \label{prop:Melons}
A melonic move adds two vertices and $q-1$ faces. The number of faces of melonic graphs is
\begin{equation}
F(G) = q + (q-1)\frac{V(G)-2}{2}
\end{equation}
hence $\delta(G)=1$ for melonic graphs.
\end{proposition}

\begin{proof}
The first statement is trivial from the definition of the melonic move. The number of faces is then obtained by induction. Indeed, $F(G_{\min}) = q$ at $V(G)=2$ for $G_{\min}$ the only melonic graph with two vertices. The induction is completed by using the first statement. $\delta(G)=1$ follows from the expression of $\delta(G)$ in \eqref{scaling}. \qed
\end{proof}

Melonic graphs satisfy a gluing rule which generalizes the melonic move (and originates from it obviously).

\begin{proposition} \label{prop:MelonicGluing}
Let $G_1, G_2\in\mathbbm{G}$ be two melonic graphs and $e_1$ in $G_1$, $e_2$ in $G_2$ two fermionic lines. Cut open $G_1$ on $e_1$ to get a 2-point function $G_1^{(e_1)}$ and similarly with $G_2$. There are two ways to glue $G_1^{(e_1)}$ and $G_2^{(e_2)}$. The resulting graphs, $G_1^{(e_1)} \star G_2^{(e_2)}$ and $G_1^{(e_1)} \bar{\star} G_2^{(e_2)}$, are melonic.
\end{proposition}

\begin{proof}
It is sufficient to consider one of the two gluings, say $G_1^{(e_1)} \star G_2^{(e_2)}$. This is proved by induction on the number of vertices of $G_1$. If $G_1$ has two vertices, $G_1 = G_{\min}$ and the insertion of $G_1^{(e_1)}$ is the melonic move on $e_2$. 

Assume the proposition holds for graphs $G_1$ with $V-2$ vertices and consider a new melonic $G_1$ with $V$ vertices. It is obtained from a melonic move performed on the melonic graph $G_1'$ on the fermionic line $e_1'$. The idea is then to find the line $e_1$ in $G_1'$, form $G_1^{'(e_1)}\star G_2^{(e_2)}$ which is melonic from the induction hypothesis and then perform the melonic move on $e_1'$ to get $G_1^{(e_1)}\star G_2^{(e_2)}$ which will thus be melonic too. This is summarized in the following commutative diagram
\begin{equation}
\begin{array}{c} \includegraphics[scale=.4]{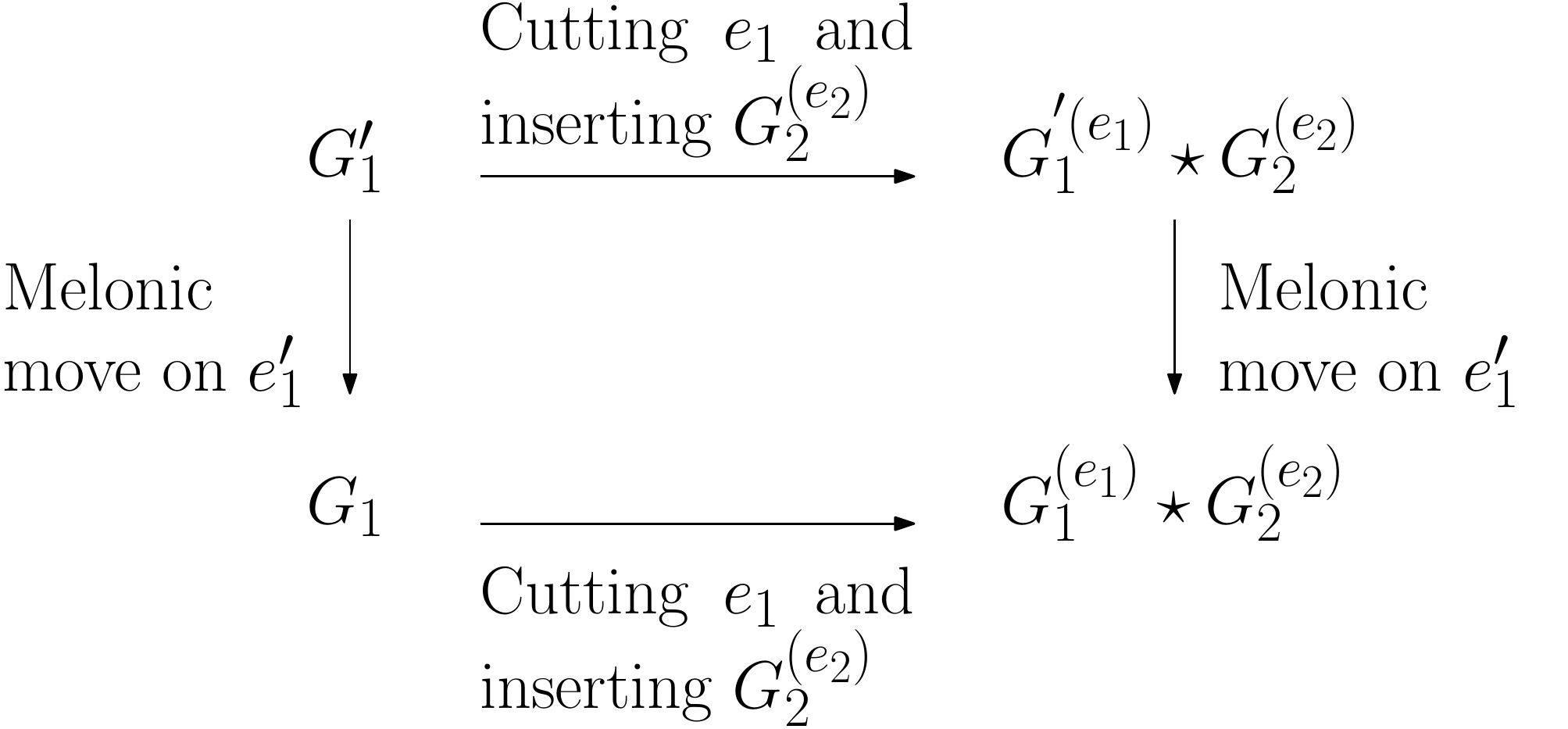} \end{array}
\end{equation}
We thus want to use the path from $G_1'$ to $G_1^{(e_1)}\star G_2^{(e_2)}$ which goes right then down. To do so, since $e_1$ is defined in $G_1$ and $e_1'$ in $G_1'$, we have to identify the equivalent lines, $e_1$ in $G_1'$ on which to open it, and $e_1'$ in $G_1^{'(e_1)}\star G_2^{(e_2)}$ on which to perform a melonic insertion.

Notice that when the melonic move is performed on $e_1'$, the latter is split into two, $e_{1L}$ and $e_{1R}$ on each side of the melonic insertion in $G_1$.

If $e_1 \neq e_{1L}, e_{1R}$, then there is a well-identified $e_1$ in $G_1'$ such that performing the melonic move on $e_1'$ in the 2-point graph $G_1^{'(e_1)}$ gives $G_1^{(e_1)}$. Notice that $e_1'$ is identified in $G_1^{'(e_1)}$ by trivial inclusion. We can then consider $G_1^{'(e_1)}\star G_2^{(e_2)}$ to be melonic and perform the melonic move on $e_1'$.

If $e_1 = e_{1L}$ (or $e_{1R}$), it means that we can identify $e_1$ with $e_1'$ in $G_1'$ and reproduce the above reasoning.\qed
\end{proof}

\subsection{2-cuts} \label{sec:2Cuts}

We recall that a 2-cut is a pair of edges in a graph whose removal (or equivalently cutting) disconnects the graph. We also recall that following \eqref{scaling} we are looking for the graphs which maximize the number of faces at fixed number of vertices. Let us denote the maximal number of faces on $V$ vertices 
\begin{equation}
F_{\max}(V) = \max_{\{G\in\mathbbm{G}, V(G) =V\}} F(G)
\end{equation}
and the set of graphs maximizing $F(G)$ at fixed $V$
\begin{equation}
\mathbbm{G}_{\max}(V) = \left\{ G\in\mathbbm{G}, \text{s.t.} \quad V(G) = V \quad \text{and} \quad F(G) = F_{\max}(V) \right\}
\end{equation}

\begin{proposition} \label{prop:2Cut}
Let $G\in\mathbbm{G}$ and $e_1, e_2$ two fermionic lines in $G$ which belong to the same face. If $\{e_1, e_2\}$ is not a 2-cut in $G$, then $G\not\in \mathbbm{G}_{\max}(V(G))$.
\end{proposition}

In other words, if there exist two lines in the same face which do not form a 2-cut, the graph is not dominant at large $N$.

\begin{proof}
There are two cases to distinguish: whether $\{e_1, e_2\}$ is a 2-cut or not in $G_0$.
\begin{description}[style=wide]
\item[$\{e_1, e_2\}$ is not a 2-cut in $G_0$.]
We draw $G$ as
\begin{equation}
G = \begin{array}{c} \includegraphics[scale=.7]{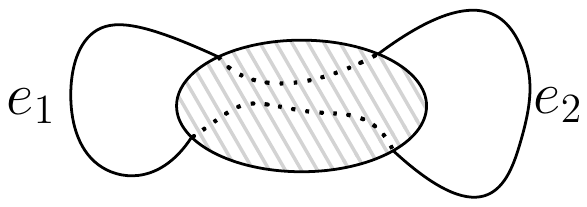} \end{array}
\end{equation}
where the dotted lines represent the paths alternating fermionic lines and strands of disorder lines which constitute the face $e_1$ and $e_2$ belongs to.

Now consider $G'$ obtained by cutting $e_1$ and $e_2$ and regluing the half-lines in the unique way which creates one additional face,
\begin{equation}
G' = \begin{array}{c} \includegraphics[scale=.7]{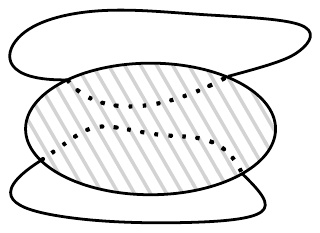} \end{array}
\end{equation}
$G_0'$ is connected since $\{e_1, e_2\}$ is not a 2-cut in $G_0$, and hence $G'\in\mathbbm{G}$. No other faces of $G$ are affected. Therefore $F(G') = F(G)+1$ and thus $G\not\in \mathbbm{G}_{\max}(V(G))$.

\item[$\{e_1, e_2\}$ is not a 2-cut in $G_0$.] 
An example of this situation is when $G_0$ is melonic but $G$ is not because the disorder lines are added in a way which does not respect melonicity.

In this case, $G$ looks like
\begin{equation}
G = \begin{array}{c} \includegraphics[scale=.7]{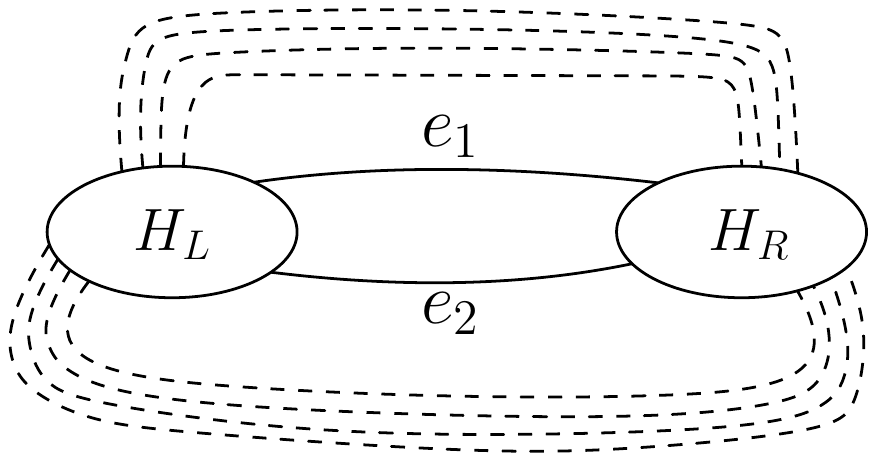} \end{array}
\end{equation}
i.e. $H_L$ and $H_R$ are both connected, and the only lines between them are $e_1, e_2$ and some disorder lines. Consider $G'$ obtained by cutting $e_1$ and $e_2$ and regluing the half-lines as follows
\begin{equation}
G' = \begin{array}{c} \includegraphics[scale=.7]{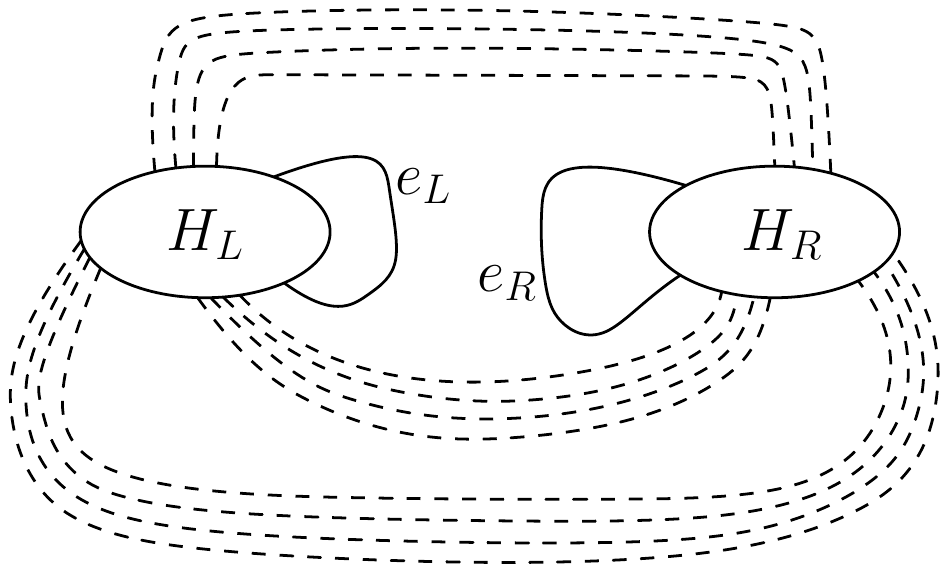} \end{array}
\end{equation}
Notice that $G'\not\in\mathbbm{G}$ since $G'_0$ consists of two connected components $G_{0L}'$ and $G_{0R}'$.

Consider a disorder line $e_0$ between them. It joins two vertices $v_L$ in $G_{0L}'$ and $v_R$ in $G_{0R}'$. We perform the contraction of the disorder line $e_0$ as follows
\begin{equation}
\begin{aligned}
&G' = \begin{array}{c} \includegraphics[scale=.6]{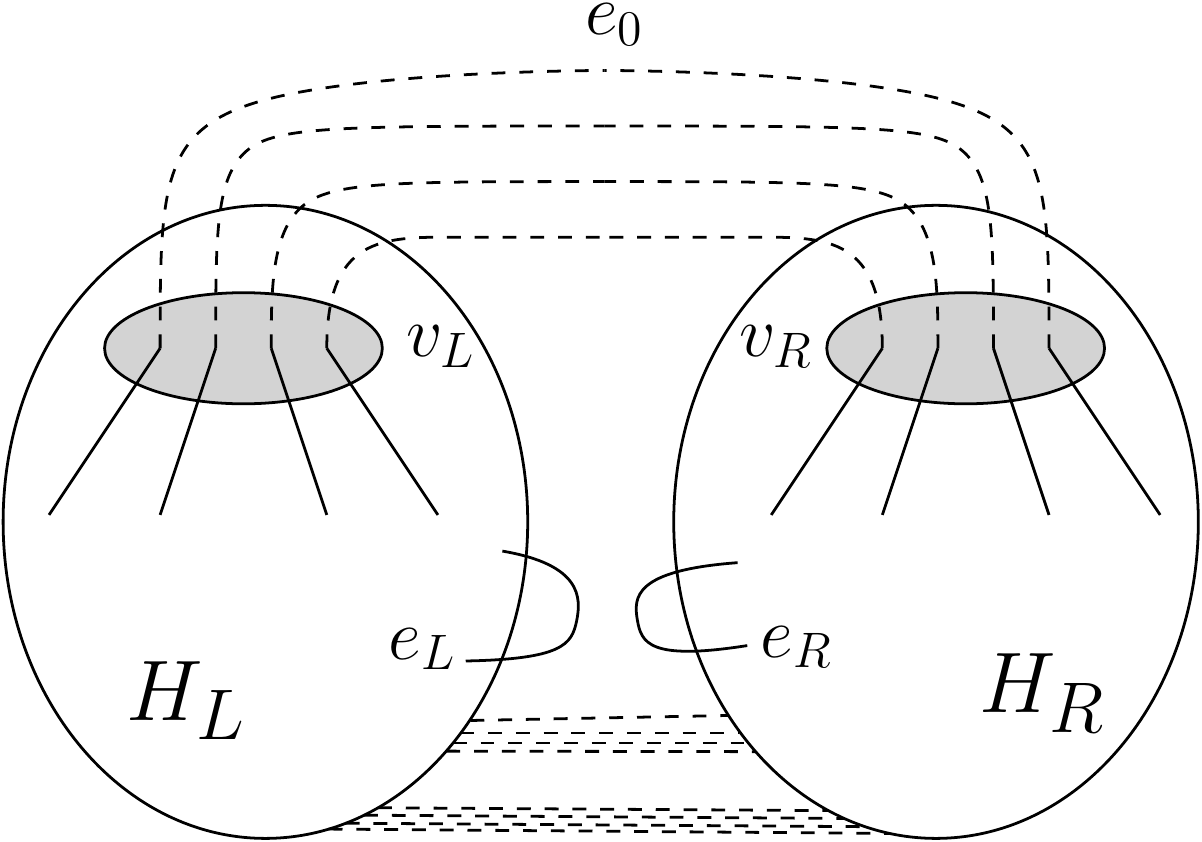} \end{array}\\
\to\quad &G'' = \begin{array}{c} \includegraphics[scale=.6]{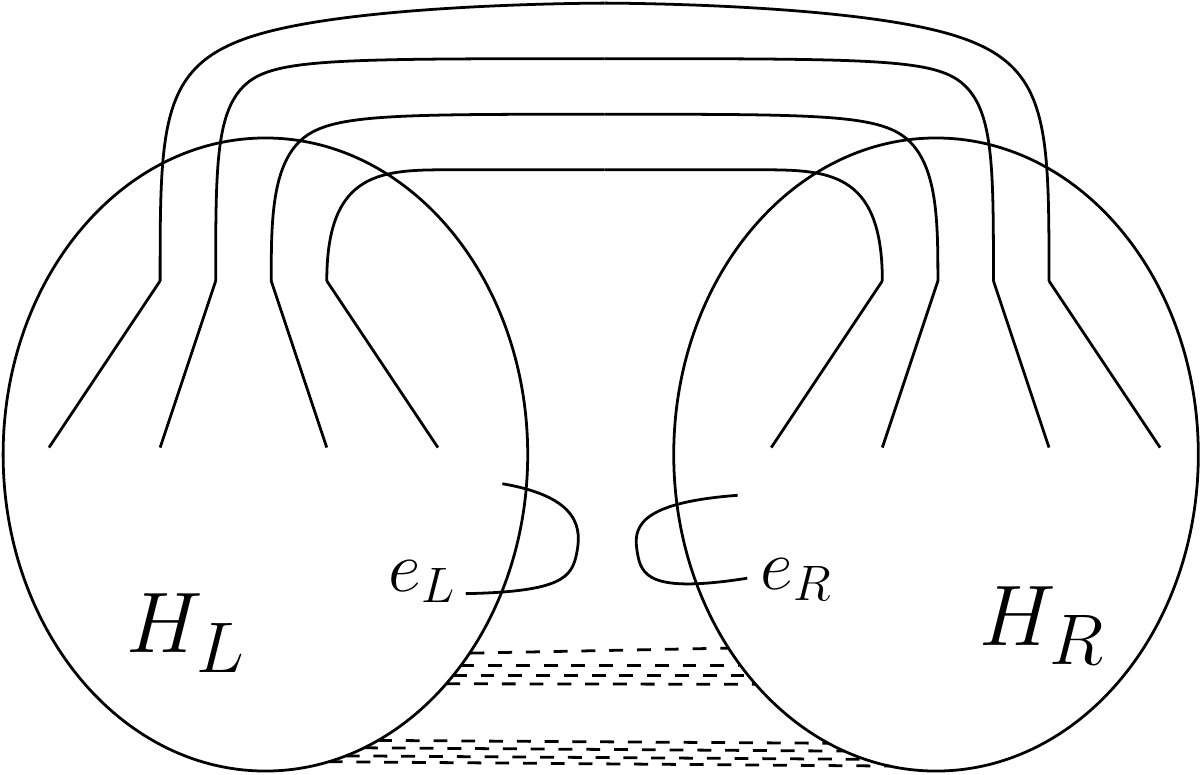} \end{array}
\end{aligned}
\end{equation}
It removes $v_L, v_R$ and $e_0$ and joins the pending fermionic lines which were connected by the strands of $e_0$. The key point is that $G''\in\mathbbm{G}$ now since the contraction of $e_0$ connects the two disjoint components of $G_0'$ by $q$ fermionic lines. 

Let us now analyze the variations of the number of faces from $G$ to $G''$. First from $G$ to $G'$: in $G$ the lines $e_1, e_2$ belong to the same face, while $e_L$ and $e_R$ may or may not belong to the same face in $G'$, hence 
\begin{equation}
F(G) \leq F(G').
\end{equation}
Then the contraction of $e_0$ does not change the number of faces. Here for this to be true, it is key that that $e_0$ connects two disjoint components of $G_0'$. Therefore $F(G) \leq F(G'')$.

To conclude the proof, notice that $G''$ has two vertices less than $G$. Therefore we can perform a melonic insertion on any fermionic line of $G''$ to get a graph $\tilde{G}\in\mathbbm{G}$ with $V(G) = V(\tilde{G})$ and crucially 
\begin{equation}
F(\tilde{G}) = F(G'') + q-1
\end{equation}
as in Proposition \ref{prop:Melons}. For $q>1$ it comes that $F(G)< F(\tilde{G})$ and thus $G\not\in\mathbbm{G}_{\max}(V(G))$.
\end{description}
\qed
\end{proof}

\subsection{Large $N$ limit} \label{sec:Thm}

We now prove that the graphs which are dominant in the large $N$ limit of the SYK model are the melonic graphs.

\begin{theorem} \label{thm}
The weight of $G\in\mathbbm{G}$ is bounded as
\begin{equation}
\delta(G) \leq 1
\end{equation}
Moreover, the graphs such that $\delta(G) = 1$ are the melonic graphs.
\end{theorem}

Using the notation of Proposition \ref{prop:2Cut}, this gives
\begin{equation}
\bigcup_{\text{$V$ even}} \mathbbm{G}_{\max}(V) = \left\{ G\in\mathbbm{G}, \text{s.t.} \quad \delta(G)=1 \right\} = \left\{ \text{Melonic graphs}\right\}.
\end{equation}

\begin{proof}
We proceed by induction.

First, $G_{\min}$ is melonic by definition. It has $F(G_{\min})=q$ and $V(G_{\min})=2$ hence satisfies $\delta(G_{\min}) = 1$. Since it is the only graph on two vertices, the theorem indeed holds on two vertices.

Let $V\geq 4$ even. We assume the theorem is true up to $V-2$ vertices and consider $G\in\mathbbm{G}$ with $V(G) = V$ vertices. Due to Proposition \ref{prop:2Cut}, we know that all pairs $\{e_1, e_2\}$ with $e_1, e_2$ two fermionic lines belonging in a common face must be 2-cuts in $G$.

Then any such pair $\{e_1, e_2\}$ also is a 2-cut in $G_0$ and $G$ takes the form
\begin{equation}
G = \begin{array}{c} \includegraphics[scale=.6]{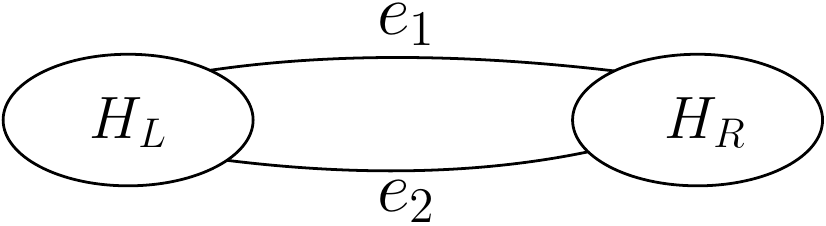} \end{array}
\end{equation}
where $H_L, H_R$ are connected, 2-point graphs. We cut $e_1$ and $e_2$ and glue the resulting half-lines to close $H_L$ and $H_R$ into $G_L, G_R$,
\begin{equation}
G_L = \begin{array}{c} \includegraphics[scale=.6]{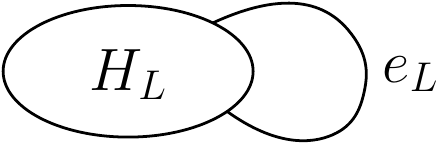} \end{array} \qquad G_R = \begin{array}{c} \includegraphics[scale=.6]{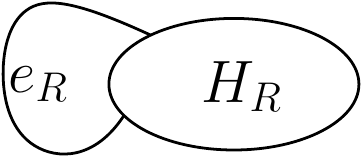} \end{array}
\end{equation}
In the notations of Proposition \ref{prop:MelonicGluing}, we have $H_L = G_L^{(e_L)}$ and $H_R = G_R^{(e_R)}$. Moreover
\begin{equation}
G = H_L \star H_R.
\end{equation}

Since $e_1$ and $e_2$ are in the same face, we have 
\begin{equation}
F(G) = F(G_L) + F(G_R) - 1,
\end{equation}
and $F(G)$ is maximal if and only if $F(G_L)$ and $F(G_R)$ are. From the induction hypothesis, this requires $G_L$ and $G_R$ to be melonic. Then $G = H_L \star H_R$ is melonic too according to Proposition \ref{prop:MelonicGluing}. \qed
\end{proof}

We get as a corollary of Theorem \ref{thm} and Proposition \ref{prop:2Cut} a characterization of melonic graphs.
\begin{corollary}
$G\in\mathbbm{G}$ is melonic if and only if all pairs $\{e_1, e_2\}$ of fermionic lines which belong in a common face are 2-cuts.
\end{corollary}


\bibliographystyle{spmpsci}      
\bibliography{MelonsSYK}  
\end{document}